\documentclass{amsart}
\usepackage{multirow}
\usepackage{multicol}
\usepackage[dvipsnames,svgnames,table]{xcolor}
\usepackage{graphicx}
\usepackage{epstopdf}
\usepackage{ulem}
\usepackage{hyperref}
\usepackage{amsmath}
\usepackage{amssymb}
\theoremstyle{plain}
\newtheorem{thm}{Theorem}[section]
\newtheorem{lem}{Lemma}[section]

\newtheorem{cor}{Corollary}[section]
\theoremstyle{definition}

\theoremstyle{remark}

\title[\resizebox{4.5in}{!}{On The Space-Time Fractional Schr\"{o}dinger Equation with time independent potentials}]{On The Space-Time Fractional Schr\"{o}dinger Equation with time independent potentials}
\author{saleh Baqer}
\thanks{S. Baqer: Corresponding author.}
\address{Department of Mathematics\\
  Faculty of Science\\
  Kuwait University}
\email[S.~Baqer]{mpcosmo57@gmail.com}
\author{Lyubomir Boyadjiev}
\address{Department of Mathematics\\
  Faculty of Science\\
  Kuwait University}
\email[L.~Boyadjiev]{boyadjievl@yahoo.com}

\begin{document}
\maketitle
\begin{abstract}This paper is about the fractional Schr\"{o}dinger equation (FSE) expressed in terms of the quantum Riesz-Feller space fractional and the Caputo time fractional derivatives. The main focus is on the case of time independent potential fields as a Dirac-delta potential and a linear potential. For such type of potential fields the separation of variables method allows to split the FSE into space fractional equation and time fractional one. The results obtained in this paper contain as particular cases already known results for FSE in terms of the quantum Riesz space fractional derivative and standard Laplace operator.
\end{abstract}

\section{Problem Formulation}
The particle motion in quantum mechanics is governed by the Schr\"{o}dinger equation (SE) and the particle motion state is determined by the wave function. Feynman and Hibbs (\cite{feynman}) deduced the standard SE by means of path integrals over the Brownian trajectories. This idea was generalized by Laskin (\cite{laskin1}, \cite{laskin2}) who introduced the principles of the fractional quantum mechanics through the path integral approach over the L\'{e}vy trajectories. For a particle with a mass $m$ moving in a potential field $V(x,t)$ Laskin introduced the FSE of the form $\left(1<\alpha\leq2\right)$,
\begin{equation}\label{1}
{-D}_{\alpha{}}{(\hslash{}\nabla{})}^{\alpha{}}\psi(x,t)+V(x,t)\psi(x,t)=i\hslash{}\frac{\partial{}}{\partial{}t}\psi(x,t).
\end{equation}
In (\ref{1}), $\hslash$ is the reduced Planck constant, ${D}_{\alpha{}}$ is the generalized fractional diffusion coefficient with physical dimension $\left[D_{\alpha{}}\right]=\text{erg}^{1-\alpha{}}\times{}\text{cm}^{\alpha{}}\times{}\text{sec}^{-\alpha{}}$ ($D_{\alpha{}}$=$1/2m$ for $\alpha{}=2$), $ \psi{}\left(x,t\right) $ is the wave function and $-{(\hslash{}\nabla{})}^{\alpha{}}$ is the quantum Riesz space fractional derivative which is defined by (\cite{laskin2})
\begin{equation}\label{qrfd}
\mathcal{F}\left\{-{(\hslash{}\nabla{})}^{\alpha{}}\psi(x,t); p\right\}={\left|p\right|}^{\alpha}\hat{\psi}(p,t),
\end{equation}
where
\begin{equation}\label{ft}
\hat{\psi}(p,t)=\mathcal{F}\left\{\psi(x,t);p\right\}=\frac{1}{2\pi{}\hslash{}}\int_{-\infty{}}^{\infty{}}e^{-ipx/\hslash{}}\
\psi(x,t)\ dx
\end{equation}
is the Fourier transform in momentum representation of a function $ \psi\in{S} $, where $ S $ is the space of rapidly decreasing functions, and
\begin{equation}\label{ift}
\psi(x,t)=\mathcal{F}^{-1}\left\{\
\hat{\psi}(p,t);x\right\}=\frac{1}{2\pi\hslash{}}\int_{-\infty{}}^{\infty{}}e^{ipx/\hslash{}}\
\hat{\psi}(p,t)\ dp
\end{equation}
is the inverse Fourier transform. The adopted Fourier transform (\ref{ft}) and its inverse (\ref{ift}) satisfy the Plancherel theorem (see e.g. \cite[Sec. 2.7]{qm}).

The FSE (\ref{1}) is solved in case of the infinite potential well, a Dirac-delta potential, a linear potential and a Coulomb potential (see \cite{laskin2}, \cite{dong1}). For $ \alpha=2 $, the equation (\ref{1}) reduces to the standard SE.

The generalization of (\ref{1}) when the time derivative is replaced by the Caputo fractional derivative (\cite{pd}, \cite{caputo}) of order $ \beta $,
\begin{equation}\label{cfd}
{{}_a^CD}_t^{\beta{}}f\left(t\right)=\left\{\begin{array}{ll}
\frac{1}{\Gamma(n-\beta{})}\
\int_a^t\frac{f^{\left(n\right)}(\tau{})}{{(t-\tau{})}^{\beta{}+1-n}}\ d\tau{}, & n-1<\beta{}<n,\ n\in{}\mathbb{N} \\
\frac{d^n}{dt^n}f\left(t\right), & \beta{}=n\in{}\mathbb{N}
\end{array} \right.,
\end{equation}
was studied recently by many authors (see e.g. \cite{dong2}, \cite{bayin1}). The FSE with the quantum Riesz-Feller space fractional derivative $ {D}_{\theta}^{\alpha} $ of order $ \alpha $, skewness $ \theta $ and the Caputo time fractional derivative $ {{}_0^CD}_{t}^{\beta} $ of order $ \beta$ has the form
\begin{equation}\label{sfse}
{C}_{\alpha}{D}^{\alpha}_{\theta}\psi(x,t)+ V(x,t)\psi(x,t)=\left(i\hslash\right)^{\beta}{{}_0^CD}_t^{\beta{}}\psi(x,t),
\end{equation}
where $ {C}_{\alpha} $ is a positive constant $ \left({C}_{\alpha}={\hslash}^{2}/2m~\text{for}~\alpha=2\right) $.
The Riesz-Feller space fractional derivative $ {{}_xD}^{\alpha}_{\theta} $, which is the generalization of the Riesz space fractional derivative, of order $ \alpha $ and skewness $ \theta $ is defined through its Fourier transform as (\cite[Def. 6.7.]{saxena})
\begin{equation}
\mathcal{F}\left\{{}_x{D}_{\theta}^{\alpha}\psi(x,t);p\right\}=-{\eta{}}_{\alpha{}}^{\theta{}}\hat{\psi}(p,t),
\end{equation}
where
\begin{equation}\label{rfc}
{\eta}_{\alpha}^{\theta}={\left\vert{p}\right\vert}^{\alpha}e^{iSgn\left(p\right)\theta\pi/2},~~~~~~~~0<\alpha\leq{2},~~~~~~~\left|\theta\right|\leq{}\text{min}\left\{\alpha,2-\alpha\right\}.
\end{equation}
 The quantum Riesz-Feller space fractional derivative $ {D}^{\alpha}_{\theta} $ is defined to be the Riesz-Feller space fractional derivative timed by a minus sign (See \cite{hello}), that is,
\[
{D}^{\alpha}_{\theta}=-{{}_xD}^{\alpha}_{\theta}.
\]
Thus for the quantum  Riesz-Feller space fractional derivative
\begin{equation}\label{qrffd}
\mathcal{F}\left\{{D}^{\alpha}_{\theta}\psi(x,t); p\right\}={\eta}_{\alpha}^{\theta}\hat{\psi}(p,t).
\end{equation}
When $ \theta=0 $, (\ref{qrffd}) reduces to
\[
\mathcal{F}\left\{{D}^{\alpha}_{0}\psi(x,t); p\right\}={\left|p\right|}^{\alpha}\hat{\psi}(p,t),
\]
where (\cite[(6.149)]{saxena})
\[
{D}_{0}^{\alpha}=-{{}_xD}_{0}^{\alpha}=-\left[-\left(-\Delta\right)^{\alpha/2}\right]=\left(-\Delta\right)^{\alpha/2}.
\]
Thus, the operator $ \left(-\Delta\right)^{\alpha/2} $ is the quantum Riesz space fractional derivative in case of generalizing the SE with the quantum Riesz-Feller space fractional derivative. For $ \theta=0,\,\alpha=2~\text{and} ~\beta=1 $ (\ref{sfse}) reduces to the standard SE. The interested readers in some new results related to the Riesz-Feller space fractional derivative, like its possible geometrical and physical meanings, can refer to the recent papers (\cite{her1}-\cite{her3}).

When $ \theta=0 $, the following FSE
\begin{equation}\label{rem1}
{C}_{\alpha}\left(-\Delta\right)^{\alpha/2}\psi(x,t)+ V(x,t)\psi(x,t)=E\psi(x,t)
\end{equation}
is equivalent to the one in terms of the fractional diffusion coefficient $ D_{\alpha} $, that is,
\begin{equation}\label{rem2}
{-D}_{\alpha{}}{(\hslash{}\nabla{})}^{\alpha{}}\psi(x,t)+V(x,t)\psi(x,t)=E\psi(x,t).
\end{equation}
Indeed, when $ \alpha=2 $ the equations (\ref{rem1}) and (\ref{rem2}) reduce to the standard SE.

For a free particle the equation (\ref{sfse}) with the standard first order time derivative is solved and the solution is obtained (\cite{hello}) in terms of the Fox $H-$function $ H^{m,n}_{p,q}(z) $ (\cite[Sec. 1.2]{saxena}).

We consider in this paper two cases of time independent potentials: a Dirac-delta potential and a linear potential. The method of separation of variables can be effectively used and under the assumption $ \psi(x,t)=f(t)\phi(x) $, to reduce equation (\ref{sfse}) to the time fractional equation
\begin{equation}\label{tfse}
\left(i\hslash\right)^{\beta}{{}_0^CD}_t^{\beta{}}f(t)=Ef(t),
\end{equation}
and the space fractional equation
\begin{equation}\label{2sfse}
C_{\alpha{}}D_{\theta{}}^{\alpha{}}\phi(x)+V(x)\phi(x)=E\phi(x),
\end{equation}
where $ E $ refers to the energy.

\section{Time Fractional Equation}

\begin{thm}(\cite{dong2})
If $ 0<\beta\leq{1} $, the solution of the time fractional equation (\ref{tfse}) is of the form
\[
f(t)=f(0)H^{1,1}_{1,2}\left[-{\left(\dfrac{t}{i\hslash}\right)}^{\beta}E\left|\begin{array}{
ll}
(0,1)\\
(0,1),(0,\beta)
\end{array}\right.\right].
\]
\end{thm}
The proof of Theorem 2.1. is based on the usage of the Laplace transform applied to (\ref{tfse}) which implies (\cite[(2.253)]{pd})
\begin{equation}\label{miss11}
F(s)=f(0)\dfrac{s^{\beta-1}}{s^{\beta}-{\left(i\hslash\right)}^{-\beta}E}.
\end{equation}
Since
\[
\frac{1}{1-{\left(i\hslash{}\right)}^{-\beta{}}Es^{-\beta{}}}=\sum_{k=0}^{\infty{}}E^k{\left(i\hslash{}\right)}^{-\beta{}k}s^{-\beta{}k},
\]
then (\ref{miss11}) becomes
\begin{equation}\label{miss2}
F\left(s\right)=f\left(0\right)\sum_{k=0}^{\infty{}}\frac{{E^k\left(i\hslash{}\right)}^{-\beta{}k}}{s^{\beta{}k+1}}.
\end{equation}
By the inverse Laplace transform applied to (\ref{miss2}) it can be seen that
\begin{equation}
f\left(t\right)=f\left(0\right)\sum_{k=0}^{\infty{}}\frac{{\left({E\left(i\hslash{}\right)}^{-\beta{}}t^{\beta{}}\right)}^k}{\Gamma\left(\beta{}k+1\right)}=f\left(0\right)E_{\beta{}}\left({\left(\frac{t}{i\hslash{}}\right)}^{\beta{}}E\right),
\end{equation}
where $ E_{\beta{}}(z) $ is the one-parameter Mittag-Leffler function (\cite[(1.55)]{pd}). It is known that (\cite[(1.135)]{saxena})
\[
E_{\beta{}}\left(z\right)=H_{1,2}^{1,1}\left(-z\left|\begin{array}{
ll}
\left(0,1\right)…… \\
\left(0,1\right),(0,\beta{})
\end{array}\right.\right),
\]
and thus the validity of the theorem is proved.
\qed

\begin{cor}(\cite{griffiths})
For $ \beta=1 $, the solution of the time fractional equation (\ref{tfse}) is of the form
\[
f\left(t\right)=f\left(0\right)e^{-iEt/\hslash}.
\]
\end{cor}
\section{Space Fractional Equation}

\begin{lem}
For any $ x,\alpha>0, \rho\geq{0} $, and $ b\in{\mathbb{C}-\left\{0\right\}} $ with $ \left|\arg(b)\right|<\pi $, the following identical formula of the Fox $ H- $function holds,
\[
\frac{x^{\rho{}}}{1+bx^{\alpha{}}}=b^{-\frac{\rho{}}{\alpha{}}}H_{1,1}^{1,1}\left(bx^{\alpha{}}\left|\begin{array}{
cc}
(\frac{\rho{}}{\alpha{}},1) \\
(\frac{\rho{}}{\alpha{}},1)
\end{array}\right.\right).
\]
\end{lem}
\begin{proof}
According to (\cite[(C.12)]{virginia}),
\begin{equation}
\frac{x^{\rho{}}}{1+bx^{\alpha{}}}=b^{-\frac{\rho{}}{\alpha{}}}G_{1,1}^{1,1}\left(bx^{\alpha{}}\left|\begin{array}{
cc}
\frac{\rho{}}{\alpha{}} \\
\frac{\rho{}}{\alpha{}}
\end{array}\right.\right),
\end{equation}
where $ G^{m,n}_{p,q}(z) $ is the $ G- $function. Since a G-function and a Fox H-function are related by
\[
G^{m,n}_{p,q}\left(z\left|\begin{array}{
cc}
a_{p} \\
b_{q}
\end{array}\right.\right)=H^{m,n}_{p,q}\left(z\left|\begin{array}{
cc}
(a_{p},1) \\
(b_{q},1)
\end{array}\right.\right),
\]
then the identical formula follows. The condition on the argument $ b $, $ \left|\arg(b)\right|<\pi $, provides the existence of the Fox $ H- $function (\cite[(1.20)]{saxena}) which is
\[
\text{if}~\sigma>0~\text{and}~\left|\arg(z)\right|<\dfrac{\pi \sigma}{2},~\text{then the Fox}~H-\text{function}~\text{exists for all}~z\neq{0}
\]
where
\[
\sigma{}=\sum_{j=1}^nA_j-\sum_{j=n+1}^pA_j+\sum_{j=1}^mB_j-\sum_{j=m+1}^qB_j.
\]
If $z=bx^{\alpha{}},~z\neq{0}$ as $ x\in{(0,\infty)} $ and $ b\in{\mathbb{C}-\left\{0\right\}} $, then $ \left|\arg(z)\right|=\left|\arg(bx^{\alpha{}})\right|=\left|\arg(b)\right| $. In our case $\sigma{}=2>0$, and therefore $ \left|\arg(b)\right|<\pi $.

\end{proof}

\subsection{Dirac-delta Potential}

Suppose a particle is moving in a Dirac-delta potential $ V(x)=-\gamma\delta(x)(\gamma>0) $. Then the space fractional equation (\ref{2sfse}) takes the form
\begin{equation}\label{dfse}
C_{\alpha{}}D_{\theta{}}^{\alpha{}}\phi(x)-\gamma{}\delta{}(x)\phi(x)=E\phi(x).
\end{equation}
\begin{thm}
If $ 1<\alpha\leq{2} $ and $ \left| \theta \right| \leq{\min\left\{\alpha,2-\alpha\right\}} $, then for $ x\neq{0} $ the solution of the equation (\ref{dfse}) has the form
\begin{multline*}
\phi(x)=
\frac{-\pi{}\gamma{}k}{{\left(2\pi{}\hslash{}\right)}^2\alpha{}E}{\left(\frac{C_{\alpha{}}}{-E}\right)}^{\frac{-1}{\alpha{}}}\times\\
\left\{e^{\frac{-i\theta{}\pi{}}{2\alpha{}}}H_{2,3}^{2,1}\left[{\left\vert{}x\right\vert{}\left(\frac{{\hslash{}}^{\alpha{}}e^{i\theta{}\pi{}/2}C_{\alpha{}}}{-E}\right)}^{-1/\alpha{}}\left|\begin{array}{
ll}
\left(\frac{\alpha{}-1}{\alpha{}},\frac{1}{\alpha{}}\right),\left(\frac{1}{2},\frac{1}{2}\right)
\\
\left(0,1\right),\left(\frac{\alpha{}-1}{\alpha{}},\frac{1}{\alpha{}}\right),\left(\frac{1}{2},\frac{1}{2}\right)
\end{array}\right.\right]\right.+\\
\left.e^{\frac{i\theta{}\pi{}}{2\alpha{}}}H_{2,3}^{2,1}\left[{\left\vert{}x\right\vert{}\left(\frac{{\hslash{}}^{\alpha{}}e^{-i\theta{}\pi{}/2}C_{\alpha{}}}{-E}\right)}^{-1/\alpha{}}\left|\begin{array}{
ll}
\left(\frac{\alpha{}-1}{\alpha{}},\frac{1}{\alpha{}}\right),\left(\frac{1}{2},\frac{1}{2}\right)
\\
\left(0,1\right),\left(\frac{\alpha{}-1}{\alpha{}},\frac{1}{\alpha{}}\right),\left(\frac{1}{2},\frac{1}{2}\right)
\end{array}\right.\right]\right\}-i\frac{\gamma{}k\sqrt{\pi{}}}{{2\left(2\pi{}\hslash{}\right)}^2\alpha{}E}\times\\
{\left(\frac{C_{\alpha{}}}{-E}\right)}^{\frac{-1}{\alpha{}}}\left\{e^{\frac{-i\theta{}\pi{}}{2\alpha{}}}\right.H_{1,3}^{2,1}\left[\left\vert{}x\right\vert{}{\left(\frac{2^{\alpha{}}{\hslash{}}^{\alpha{}}C_{\alpha{}}e^{i\theta{}\pi{}/2}}{-E}\right)}^{-1/\alpha{}}\left|\begin{array}{
ll}
\left(\frac{\alpha{}-1}{\alpha{}},\frac{1}{\alpha{}}\right) \\
\left(\frac{1}{2},\frac{1}{2}\right),\left(\frac{\alpha{}-1}{\alpha{}},\frac{1}{\alpha{}}\right),\left(0,\frac{1}{2}\right)
\end{array}\right.\right]\\
\left.{-\
e^{\frac{i\theta{}\pi{}}{2\alpha{}}}H}_{1,3}^{2,1}\left[\left\vert{}x\right\vert{}{\left(\frac{2^{\alpha{}}{\hslash{}}^{\alpha{}}C_{\alpha{}}e^{-i\theta{}\pi{}/2}}{-E}\right)}^{-1/\alpha{}}\left|\begin{array}{
ll}
\left(\frac{\alpha{}-1}{\alpha{}},\frac{1}{\alpha{}}\right) \\
\left(\frac{1}{2},\frac{1}{2}\right),\left(\frac{\alpha{}-1}{\alpha{}},\frac{1}{\alpha{}}\right),\left(0,\frac{1}{2}\right)
\end{array}\right.\right]\right\},
\end{multline*}

where
\[
k=2\pi\hslash~\mathcal{F}\left\{\delta(x)\phi(x); p\right\}.
\]
\end{thm}

\begin{proof}
Taking into account (\ref{qrffd}), the application of the Fourier transform (\ref{ft}) to (\ref{dfse}) leads to
\begin{equation}
\hat{\phi}(p)=\
\frac{\gamma{}k}{2\pi{}\hslash{}}\frac{1}{\left(C_{\alpha{}}{\left\vert{}p\right\vert{}}^{\alpha{}}e^{iSgn\left(p\right)\theta{}\pi{}/2}-E\right)},
\end{equation}
and thus
\begin{equation}\label{ii}
\phi(x)=\frac{\gamma{}k}{{\left(2\pi{}\hslash{}\right)}^2}\left(I_1+iI_2\right),
\end{equation}
where
\[
I_1=\int_{-\infty{}}^{\infty{}}\frac{\cos{\left(px/\hslash{}\right)}}{\left(C_{\alpha{}}{\left\vert{}p\right\vert{}}^{\alpha{}}e^{iSgn\left(p\right)\theta{}\pi{}/2}-E\right)}\
dp,
\]
\[
I_2=\int_{-\infty{}}^{\infty{}}\frac{\sin{\left(px/\hslash{}\right)}}{\left(C_{\alpha{}}{\left\vert{}p\right\vert{}}^{\alpha{}}e^{iSgn\left(p\right)\theta{}\pi{}/2}-E\right)}\
dp.
\]
If $ x>0 $, by Lemma 3.1. it is possible to see that
\begin{multline*}
I_1=\frac{-1}{E}\left\{\int_0^{\infty{}}\
\cos{\left(px/\hslash{}\right)}H_{1,1}^{1,1}\left(-p^{\alpha{}}\frac{C_{\alpha{}}e^{-i\theta{}\pi{}/2}}{E}\left|\begin{array}{
cc}
(0,1) \\
(0,1)
\end{array}\right.\right)dp\right.\\
+\left.\int_0^{\infty{}}\cos{\left(px/\hslash{}\right)}H_{1,1}^{1,1}\left(-p^{\alpha{}}\frac{C_{\alpha{}}e^{i\theta{}\pi{}/2}}{E}\left|\begin{array}{
cc}
(0,1) \\
(0,1)
\end{array}\right.\right)\ dp\right\},
\end{multline*}
\begin{multline*}
I_2=\frac{-1}{E}\left\{\int_0^{\infty{}}\
\sin{\left(px/\hslash{}\right)}H_{1,1}^{1,1}\left(-p^{\alpha{}}\frac{C_{\alpha{}}e^{i\theta{}\pi{}/2}}{E}\left|\begin{array}{
cc}
(0,1) \\
(0,1)
\end{array}\right.\right)dp\right.\\
-\left.\int_0^{\infty{}}\sin{\left(px/\hslash{}\right)}H_{1,1}^{1,1}\left(-p^{\alpha{}}\frac{C_{\alpha{}}e^{-i\theta{}\pi{}/2}}{E}\left|\begin{array}{
cc}
(0,1) \\
(0,1)
\end{array}\right.\right)\ dp\right\}.
\end{multline*}
According to the formulas for cosine and sine transforms of the Fox $ H- $funcion (\cite[(17)]{hello2}) and (\cite[(2.49)]{saxena}), respectively, and by using the following properties of the Fox $ H- $function (\cite{saxena}):
\[
H^{m,n}_{p,q}\left(z\left|\begin{array}{c}
                            (a_{p},A_{p}) \\
                            (b_{q},B_{q})
                          \end{array}
\right.\right)=kH^{m,n}_{p,q}\left(z^{k}\left|\begin{array}{c}
                            (a_{p},kA_{p}) \\
                            (b_{q},kB_{q})
                          \end{array}
\right.\right);~~k>0,
\]
\[
H^{m,n}_{p,q}\left(z\left|\begin{array}{c}
                            (a_{p},A_{p}) \\
                            (b_{q},B_{q})
                          \end{array}
\right.\right)=H^{n,m}_{q,p}\left(\frac{1}{z}\left|\begin{array}{c}
                            (1-b_{q},B_{q}) \\
                            (1-a_{p},A_{p})
                          \end{array}
\right.\right),
\]
\[
z^{\sigma}H^{m,n}_{p,q}\left(z\left|\begin{array}{c}
                            (a_{p},A_{p}) \\
                            (b_{q},B_{q})
                          \end{array}
\right.\right)=H^{m,n}_{p,q}\left(z\left|\begin{array}{c}
                            (a_{p}+\sigma{A_{p}},A_{p}) \\
                            (b_{q}+\sigma{B_{q}},B_{q})
                          \end{array}
\right.\right);~~\sigma\in\mathbb{C},
\]
we can get
\begin{multline*}
I_1=
\frac{-\pi{}}{\alpha{}E}{\left(\frac{C_{\alpha{}}}{-E}\right)}^{\frac{-1}{\alpha{}}}\times\\
\left\{e^{\frac{-i\theta{}\pi{}}{2\alpha{}}}H_{2,3}^{2,1}\left[{x\left(\frac{{\hslash{}}^{\alpha{}}e^{i\theta{}\pi{}/2}C_{\alpha{}}}{-E}\right)}^{-1/\alpha{}}\left|\begin{array}{
ll}
\left(1-\left(1/\alpha\right),1/\alpha{}\right),\left(1/2,1/2\right)
\\
\left(0,1\right),\left(1-\left(1/\alpha\right),1/\alpha{}\right),\left(1/2,1/2\right)
\end{array}\right.\right]\right. \\
+e^{\frac{i\theta{}\pi{}}{2\alpha{}}}H_{2,3}^{2,1}\left.\left[{x\left(\frac{{\hslash{}}^{\alpha{}}e^{-i\theta{}\pi{}/2}C_{\alpha{}}}{-E}\right)}^{-1/\alpha{}}\left|\begin{array}{
ll}
\left(1-\left(1/\alpha\right),1/\alpha{}\right),\left(1/2,1/2\right)
\\
\left(0,1\right),\left(1-\left(1/\alpha\right),1/\alpha{}\right),\left(1/2,1/2\right)
\end{array}\right.\right]\right\},
\end{multline*}
and
\begin{multline*}
I_2=\frac{-\sqrt{\pi{}}}{2\alpha{}E}{\left(\frac{C_{\alpha{}}}{-E}\right)}^{\frac{-1}{\alpha{}}}\times\\
\left\{e^{\frac{-i\theta{}\pi{}}{2\alpha{}}}H_{1,3}^{2,1}\left[x{\left(\frac{2^{\alpha{}}{\hslash{}}^{\alpha{}}C_{\alpha{}}e^{i\theta{}\pi{}/2}}{-E}\right)}^{-1/\alpha{}}\left|\begin{array}{
ll}
\left(1-\left(1/\alpha\right),1/\alpha{}\right) \\
\left(1/2,1/2\right),\left(1-\left(1/\alpha\right),1/\alpha{}\right),(0,1/2)
\end{array}\right.\right]\right.\\
-\left.{e}^{\frac{i\theta{}\pi{}}{2\alpha{}}}H_{1,3}^{2,1}\left[x{\left(\frac{2^{\alpha{}}{\hslash{}}^{\alpha{}}C_{\alpha{}}e^{-i\theta{}\pi{}/2}}{-E}\right)}^{-1/\alpha{}}\left|\begin{array}{
ll}
\left(1-\left(1/\alpha\right),1/\alpha{}\right) \\
\left(1/2,1/2\right),\left(1-\left(1/\alpha\right),1/\alpha{}\right),(0,1/2)
\end{array}\right.\right]\right\}.
\end{multline*}\\
The substitution of these expressions for $ I_1 $ and $ I_2 $ into (\ref{ii}) confirms the validity of the theorem.\\
The case $ x<0 $ can be considered similarly in order to accomplish the proof.
\end{proof}

\raggedright{Taking into account (\ref{rem1}) and (\ref{rem2}) the following assertions follow.}

\begin{cor}(\cite{dong1})
If $ 1<\alpha\leq{2} $ and $ \theta=0 $, the solution of the equation (\ref{dfse}) for $ x\neq{0} $ has the form
\begin{equation}
\phi(x)={{\xi{}}_0H}_{2,3}^{2,1}\left[{\left\vert{}x\right\vert{}\left(\frac{{\hslash{}}^{\alpha{}}D_{\alpha{}}}{-E}\right)}^{-1/\alpha{}}\left|\begin{array}{
ll}
\left(\frac{\alpha{}-1}{\alpha{}},\frac{1}{\alpha{}}\right),\left(\frac{1}{2},\frac{1}{2}\right)
\\
\left(0,1\right),\left(\frac{\alpha{}-1}{\alpha{}},\frac{1}{\alpha{}}\right),\left(\frac{1}{2},\frac{1}{2}\right)
\end{array}\right.\right],
\end{equation}
where
\[
{\xi{}}_0=\left(\frac{-\gamma{}k}{2\pi{}{\hslash{}}^2E\alpha{}}\right){\left(\frac{D_{\alpha{}}}{-E}\right)}^{\frac{-1}{\alpha{}}}.\
\]
\end{cor}

\begin{cor}(\cite{griffiths})
If $ \alpha=2 $ and $ \theta=0 $ , the standard wave function solution of the equation (\ref{dfse}) for $ x\neq{0} $ has the form
\[
\phi(x)={{\lambda}H}_{0,1}^{1,0}\left(\left\vert{}x\right\vert{}\frac{\sqrt{-2mE}}{\hslash{}}\left|\begin{array}{
cc}
\\
\left(0,1\right)
\end{array}\right.\right)={\lambda}e^{-\left\vert{}x\right\vert{}{\sqrt{-2mE}}/\hslash{}}\,\,\,\,(\lambda~is~a~constant).
\]
\end{cor}
\subsection{Linear Potential}

Consider a particle in a linear potential field
\[
V(x)=\left\{\begin{array}{ll}
Ax, & x\geq{0}\left(A>0\right) \\
\infty, & x<0
\end{array} \right..
\]
Then the space fractional equation (\ref{2sfse}) becomes
\begin{equation}\label{lfse}
C_{\alpha{}}D_{\theta{}}^{\alpha{}}\phi(x)+Ax\phi(x)=E\phi(x),~x\geq{}0\ .
\end{equation}
\begin{thm}
If $ 1<\alpha\leq{2} $ and $ \left| \theta \right| \leq{\min\left\{\alpha,2-\alpha\right\}} $, then the solution of the equation (\ref{lfse}) has the form
\begin{multline*}
\,\,\,\phi(x)=\\
\frac{2\pi{}N}{\left(\alpha{}+1\right)}H_{2,2}^{1,1}\left[\left(x-\frac{E}{A}\right)\frac{1}{\hslash{}}{\left(\frac{C_{\alpha{}}}{\hslash{}A\left(\alpha{}+1\right)}\right)}^{\frac{-1}{\alpha{}+1}}\left|\begin{array}{
ll}
\left(\frac{\alpha{}}{\alpha{}+1},\frac{1}{1+\alpha{}}\right),\left(\frac{2+\alpha{}-\theta{}}{2\left(\alpha{}+1\right)},\frac{\alpha{}+\theta{}}{2\left(\alpha{}+1\right)}\right)
\\
\left(0,1\right),\left(\frac{2+\alpha{}-\theta{}}{2\left(\alpha{}+1\right)},\frac{\alpha{}+\theta{}}{2\left(\alpha{}+1\right)}\right)
\end{array}\right.\right],
\end{multline*}
where
\[
N=\frac{1}{2\pi{}\hslash{}}{\left(\frac{C_{\alpha{}}}{A\hslash{}\left(\alpha{}+1\right)}\right)}^{\frac{-1}{(\alpha{}+1)}}.
\]
\end{thm}
\begin{proof}
According to (\ref{qrffd}) and the formula (\cite{dong1})
\[
\mathcal{F}\
\left\{x\phi(x);p\right\}=i\hslash{}\frac{d}{dp}\hat{\phi}(p),
\]
the application of the Fourier transform (\ref{ft}) to the equation (\ref{lfse}) leads to
\[
\frac{d\hat{\phi}(p)}{\hat{\phi}(p)}=\frac{1}{Ai\hslash{}}\left(E-C_{\alpha{}}{\left\vert{}p\right\vert{}}^{\alpha{}}e^{iSgn(p)\theta{}\pi/2}\right),
\]
from where it follows readily that (omitting the constant of the integration)
\[
\hat{\phi}(p)=\left\{\begin{array}{l}\begin{array}{
ll}
\exp\left[\frac{-i}{A\hslash{}}\left(Ep-\frac{C_{\alpha{}}}{\alpha{}+1}p^{\alpha{}+1}e^{i\theta{}\pi/2}\right)\
\right]; & p>0
\end{array} \\
\begin{array}{
ll}
\exp\left[\frac{-i}{A\hslash{}}\left(Ep+\frac{C_{\alpha{}}}{\alpha{}+1}{\left\vert{}p\right\vert{}}^{\alpha{}+1}e^{-i\theta{}\pi/2}\right)\
\right]; &  p<0
\end{array}\end{array}.\right.
\]
Setting
\begin{equation}\label{subs}
w=p{\left(\frac{C_{\alpha{}}}{A\hslash{}\left(\alpha{}+1\right)}\right)}^{\frac{1}{(\alpha{}+1)}},~~y=\frac{1}{\hslash{}}\left(x-\frac{E}{A}\right){\left(\frac{C_{\alpha{}}}{A\hslash{}\left(\alpha{}+1\right)}\right)}^{\frac{-1}{(\alpha{}+1)}},
\end{equation}
it is possible by the application of the inverse Fourier transform (\ref{ift}) to get that
\begin{equation}\label{addphis}
\phi(y)=N\left\{{\phi{}}_1(y)+{\phi{}}_2(y)\right\},
\end{equation}
where
\[
{\phi}_1(y)=\int_0^{\infty{}}e^{iyw}e^{i{e^{i\theta\pi/2}w}^{\alpha{}+1}}\
\ dw,
\]
and
\[
{\phi}_2(y)=\int_{-\infty{}}^0\
e^{iyw}e^{-i{e^{-i\theta\pi/2}\left\vert{}w\right\vert{}}^{\alpha{}+1}}\
dw.
\]
Denote by $ \hat{\phi}(s)=\mathcal{M}\left\{\phi(y); s\right\} $ the Mellin transform of $ \phi(y) $. From the formula (\cite[Ch.8]{it})
\begin{equation}\label{mexp}
\mathcal{M}\left\{e^{-i\rho{x}};s\right\}=\left(i\rho\right)^{-s}\Gamma(s)={\rho}^{-s}\Gamma(s)\left(\cos\left(\frac{\pi{s}}{2}\right)-i\sin\left(\frac{\pi{s}}{2}\right)\right);~~\rho\in{\mathbb{C}},
\end{equation}
it follows that
\begin{equation}\label{1tilde}
{\tilde{\phi}}_1(s)={\left(-i\right)}^{-s}\Gamma\left(s\right)\int_0^{\infty{}}{e}^{ie^{i\theta\pi/2}w^{\alpha{}+1}}w^{-s}dw,
\end{equation}
and
\begin{equation}\label{2tilde}
{\tilde{\phi}}_2(s)={\left(-i\right)}^{-s}\Gamma\left(s\right)\int_{-\infty{}}^0e^{-ie^{-i\theta\pi/2}{\left\vert{}w\right\vert{}}^{\alpha{}+1}}w^{-s}dw.
\end{equation}

Now by the formula (\ref{mexp}), the substitution $ u={w}^{\alpha+1} $ in (\ref{1tilde}) and the substitutions $ u=-w $, $ \xi={u}^{\alpha+1} $ in (\ref{2tilde}), it can be seen that
\[
{\tilde{\phi}}_1(s)=\frac{1}{\alpha{}+1}\Gamma\left(\frac{1-s}{1+\alpha{}}\right)\Gamma\left(s\right)\exp\left(\frac{i\pi{}}{2}\left[\frac{1-\theta{}}{\alpha{}+1}+\frac{\left(\alpha{}+\theta{}\right)s}{\alpha{}+1}\right]\right),
\]
and
\[
{\tilde{\phi}}_2(s)=\frac{1}{\alpha{}+1}\Gamma\left(\frac{1-s}{1+\alpha{}}\right)\Gamma\left(s\right)\exp\left(\frac{-i\pi{}}{2}\left[\frac{1-\theta{}}{\alpha{}+1}+\frac{\left(\alpha{}+\theta{}\right)s}{\alpha{}+1}\right]\right).
\]
These representations and the formula
\[
\cos\left(\pi{z}/2\right)=\frac{\pi{}}{\Gamma\left(\frac{1+z}{2}\right)\Gamma\left(\frac{1-z}{2}\right)}
\]
allow from (\ref{addphis}) to be obtained that
\begin{equation}
\tilde{\phi}(s)=\frac{2\pi{}N}{\left(\alpha{}+1\right)}\frac{\Gamma\left(s\right)\Gamma\left(\frac{1-s}{1+\alpha{}}\right)}{\Gamma\left(\frac{\alpha{}+\theta{}-\alpha{}s-\theta{}s}{2\left(\alpha{}+1\right)}\right)\Gamma\left(\frac{2+\alpha{}-\theta{}+\alpha{}s+\theta{}s}{2\left(\alpha{}+1\right)}\right)}.
\end{equation}
Therefore
\begin{equation}
\phi(y)=\frac{2\pi{}N}{\left(\alpha{}+1\right)}\frac{1}{2\pi{}i}\int_{\gamma{}-i\infty{}}^{\gamma{}+i\infty{}}\
\frac{\Gamma\left(s\right)\Gamma\left(\frac{1-s}{1+\alpha{}}\right)}{\Gamma\left(\frac{\alpha{}+\theta{}-\alpha{}s-\theta{}s}{2\left(\alpha{}+1\right)}\right)\Gamma\left(\frac{2+\alpha{}-\theta{}+\alpha{}s+\theta{}s}{2\left(\alpha{}+1\right)}\right)}y^{-s}ds,
\end{equation}
that is (\cite[Sec. 1.2]{saxena}),
\[
\phi(y)=\frac{2\pi{}N}{\left(\alpha{}+1\right)}H_{2,2}^{1,1}\left(y\left|\begin{array}{
ll}
\left(\frac{\alpha{}}{\alpha{}+1},\frac{1}{1+\alpha{}}\right),\left(\frac{2+\alpha{}-\theta{}}{2\left(\alpha{}+1\right)},\frac{\alpha{}+\theta{}}{2\left(\alpha{}+1\right)}\right)
\\
\left(0,1\right),\left(\frac{2+\alpha{}-\theta{}}{2\left(\alpha{}+1\right)},\frac{\alpha{}+\theta{}}{2\left(\alpha{}+1\right)}\right)
\end{array}\right.\right).
\]
Taking into account the substitutions (\ref{subs}), the validity of the theorem follows directly.
\end{proof}

\raggedright{Having in mind (\ref{rem1}) and (\ref{rem2}) the following assertions follow.}
\begin{cor}(\cite{dong1})
If $ 1<\alpha\leq{2} $ and $ \theta=0 $, the solution of the equation (\ref{lfse}) has the form
\begin{multline}\label{miss13}
\,\,\,\phi(x)=\\
\frac{2\pi{}N}{\left(\alpha{}+1\right)}H_{2,2}^{1,1}\left[\left(x-\frac{E}{A}\right)\frac{1}{\hslash{}}{\left(\frac{D_{\alpha{}}}{\hslash{}A\left(\alpha{}+1\right)}\right)}^{\frac{-1}{\alpha{}+1}}\left|\begin{array}{
ll}
\left(\frac{\alpha{}}{\alpha{}+1},\frac{1}{1+\alpha{}}\right),\left(\frac{2+\alpha{}}{2\left(\alpha{}+1\right)},\frac{\alpha{}}{2\left(\alpha{}+1\right)}\right)
\\
\left(0,1\right),\left(\frac{2+\alpha{}}{2\left(\alpha{}+1\right)},\frac{\alpha{}}{2\left(\alpha{}+1\right)}\right)
\end{array}\right.\right].
\end{multline}
\end{cor}

By the power series representation of the Fox $ H- $function (\cite[(A.69)]{saxena}) and the formula
\[
\sin(\pi{z})=\frac{\pi}{\Gamma(z)\Gamma(1-z)},
\]
we can write (\ref{miss13}) as
\begin{multline}\label{series}
\phi(x)=\frac{2N}{\left(\alpha{}+1\right)}\times\\
\sum_{k=0}^{\infty{}}\Gamma\left(\frac{k+1}{\alpha{}+1}\right)\sin{\left(\frac{(\alpha+2)(k+1)\pi}{2(\alpha+1)}\right)}\frac{1}{k!}{\left[\left(x-\frac{E}{A}\right)\frac{1}{\hslash{}}{\left(\frac{D_{\alpha{}}}{A\hslash{}\left(\alpha{}+1\right)}\right)}^{\frac{-1}{\left(\alpha{}+1\right)}}\right]}^k.
\end{multline}

\raggedright{
This enables us to formulate the following statement.}
\begin{cor}(\cite{lifshitz})
If $ \alpha=2 $ and $ \theta=0 $, the standard wave function solution of the equation (\ref{lfse}) has the form
\[
\phi(x)=\frac{\lambda}{\pi{}}\sum_{k=0}^{\infty{}}\Gamma\left(\frac{k+1}{3}\right)\sin{\left(\frac{2\left(k+1\right)\pi{}}{3}\right)}\frac{1}{k!}{\left(3^{\frac{1}{3}}u\right)}^k,
\]
where $ \lambda $ is a constant and $ u=\left(x-\left(\frac{E}{A}\right)\right){\left(\frac{2mA}{{h}^2}\right)}^{\frac{1}{3}}. $
\end{cor}

\end{document}